\documentclass{llncs}
\usepackage{graphicx}

\usepackage{latexsym}
\usepackage{amsmath}
\usepackage{amssymb}

\def\expect{\mathbb{E}}

\newcommand{\BigTheta}{\mathrm{\Theta}}
\newcommand{\BigOmega}{\mathrm{\Omega}}
\newcommand{\Po}{\mathcal{P}}

\def\final{1}  
\def\iflong{1}

\ifnum\final=0  
\newcommand{\onote}[1]{[{\tiny neil: \bf #1}]\marginpar{*}}
\newcommand{\bnote}[1]{[{\tiny bruce: \bf #1}]\marginpar{*}}
\newcommand{\gnote}[1]{[{\tiny navin: \bf #1}]\marginpar{*}}
\newcommand{\sidecomment}[1]{\marginpar{\tiny #1}}
\else 
\newcommand{\onote}[1]{}
\newcommand{\bnote}[1]{}
\newcommand{\gnote}[1]{}
\newcommand{\sidecomment}[1]{}
\fi  

\newcommand{\optdyn}{\ensuremath{\text{\sc opt}_{FR}}}
\newcommand{\optmpr}{\ensuremath{\text{\sc opt}_{MPR}}}
\newcommand{\optspr}{\ensuremath{\text{\sc opt}_{SPR}}}
\newcommand{\optT}{\ensuremath{\text{\sc opt}_{TR}}}

\newcommand{\flow}[1]{f_{{#1}}}

\newcommand{\BaR}{{\sc bar}}
\newcommand{\SPR}{{\sc spr}}
\newcommand{\MPR}{{\sc mpr}}
\newcommand{\FR}{{\sc fr}}
\newcommand{\TR}{{\sc tr}}
\newcommand{\ssrob}{{\sc ssrob}}

\DeclareMathOperator{\cost}{cost}

\DeclareMathOperator{\supp}{supp}

\begin{document}

\title{Dynamic vs Oblivious Routing in Network Design}
\newcommand{\affaddr}[1]{#1}
\newcommand{\alignauthor}{}
\newcommand{\mcgilladdress}{
    \affaddr{Department of Mathematics and Statistics}\\
    \affaddr{McGill University, Montreal, Canada}\\
}

\author{Navin Goyal\inst{1} \and Neil Olver\inst{2} \and F. Bruce Shepherd\inst{3}}
\institute{Microsoft Reseach India\\ Bangalore, India\\ \email{navingo@microsoft.com\vspace{1ex}}
\and \mcgilladdress \email{olver@math.mcgill.ca \vspace{1ex}}
\and \mcgilladdress \email{bruce.shepherd@mcgill.ca}
}

\date{}
\maketitle

\begin{abstract}
Consider the robust network design problem of finding a minimum cost network with enough capacity to route all traffic demand matrices in a given polytope.
We investigate the impact of different routing models in this robust setting: in particular, we compare \emph{oblivious} routing, where the routing between each terminal pair must be fixed in advance,
to \emph{dynamic} routing, where routings may depend arbitrarily on the current demand.
Our main result is a construction that shows that the optimal cost of such a network
based on oblivious routing (fractional or integral) may be a factor
of $\BigOmega(\log{n})$ more than the cost required when using
dynamic routing. This is true even in the important special case of
the asymmetric hose model. This answers a question in
\cite{chekurisurvey07}, and is tight up to constant factors.  Our
proof technique builds on a connection between expander graphs and
robust design for single-sink traffic patterns~\cite{ChekuriHardness07}.
\end{abstract}

\section{Introduction}
One of the most widely studied applications of robustness in
discrete optimization has been in the context  of network design.
This is partly motivated by the fact that traffic demands in modern
data networks are often hard to determine and/or are rapidly
changing. In one general model (cf. \cite{benameur03}), the input
consists of a graph (network topology) where each edge comes with a
cost to reserve capacity.
In addition, a universe of possible demand matrices
is specified as a polyhedron $\mathcal{P}$ (or more generally, as a convex
body).  In this paper our focus
is on undirected demands and so for a demand matrix $D$, the entries
$D_{ij}$ and $D_{ji}$ normally represent the same demand, and are
hence equal.  The problem is to design a minimum cost network such
that each demand matrix in the polytope can be routed (according to
routing models we describe shortly) in the resulting capacitated
network. Typically we seek to install edge capacities so that the
sum of costs is minimized, but other cost measures such as
minimizing the maximum congestion are also considered in the
literature.  We refer to the recent survey by Chekuri \cite{chekurisurvey07}
for a discussion of these models and previous work.

Since demands are potentially changing, there are two prime natural
routing models that are considered. The first is \emph{dynamic routing}: for any given demand $D \in \mathcal{P}$, we may use a traffic routing tailored to this demand.
We consider only the case where the routing may be an arbitrary multicommodity flow, i.e. traffic flows may be fractional.
We also refer to this routing model as \FR.
Dynamic routing, out of all possible routing models, clearly leads to the cheapest possible solution.
However, this model is typically considered impractical.


On the other extreme, \emph{oblivious routing} models, inspired by routing in packet networks, ask for a \emph{routing template} that defines ahead of
time how any future demands will be routed. For each node pair
$i,j$, the template $f$ specifies a unit network flow $\flow{ij}$
between $i$ and $j$. The interpretation is that if there is a future
demand of $D_{ij}$ between nodes $i,j$, then along each $ij$ path
$P$, we should route $D_{ij} \flow{ij}(P)$ flow on this path.
This scheme is obviously much simpler than dynamic routing, and has the advantage that routings are stable, which can be important in maintaining Quality of Service guarantees.

Flow templates may be either fractional, in which case they are called
\emph{multipath routings} (\MPR), or integral, in which case they are called \emph{single-path routings} (\SPR). We also discuss a special case of \SPR\ templates
called \emph{tree} templates where the support of $f$ induces a tree
in the network; we refer to this model as \TR.  We can now formally
define the robust network design problem (cf.
\cite{ChekuriHardness07}):

\begin{definition}
Given a graph $G=(V,E)$ with $|V| = n$, edge costs $c: E \rightarrow \mathbb{R}^+$, a polytope
$\mathcal{P}$ of demand matrices, and a routing model (\FR, \SPR, \MPR, \TR), the robust network
design problem is defined as follows.  Find a minimum cost capacity installation
of edge capacities $u: E \rightarrow \mathbb{R}^+$ so that all demand matrices in $\mathcal{P}$
can be routed in the given routing model.  The cost of capacity installation $u$ is given
by $\sum_{e \in E} u(e) c(e)$.
\end{definition}

For a given instance of robust network design $(G, c, \mathcal{P})$, we use
$\optdyn(G, c, \mathcal{P})$, $\optmpr(G, c, \mathcal{P})$,
$\optspr(G, c, \mathcal{P})$ and $\optT(G, c, \mathcal{P})$ to denote the corresponding cost of
an optimally designed robust network for the four routing models. If
the context is clear, we may simply write, for instance, $\optdyn$.

Obviously we have
\begin{equation}\label{eq:gaps}
	\optdyn \leq \optmpr \leq \optspr \leq \optT.
\end{equation}
It was already known that the gap between $\optdyn$ and $\optspr$ is $O(\log n)$ (credited to A. Gupta, cf. \cite{chekurisurvey07}).
This follows by an application of the approximation of
arbitrary metrics by tree metrics~\cite{FRT04}.  One can further show, by similar arguments
but now using a theorem of \cite{AbrahamBN08} instead, that the gap between $\optdyn$ and $\optT$ is
at most $\tilde{O}(\log{n})$, where $\tilde{O}$ hides an $O(\mathrm{poly}\log\log{n})$ factor.

\paragraph{Our Results.} In this paper, we seek to understand to what extent these gaps are realizable; in
other words, for any pair of routing methods, what is the maximum possible gap between the costs of their
optimal solution?

In short, the answer is that except for the pair $\{\optmpr, \optspr\}$, the gap between any pair in~\eqref{eq:gaps} can be as large as $\BigOmega(\log n)$; this is essentially tight.
The exception, the gap between $\optmpr$ and $\optspr$, is at least polylogarithmically large ($\BigOmega(\log^{1/4-\epsilon}(n))$ for any $\epsilon > 0$).
This follows indirectly from an approximation-preserving reduction \cite{OlverS10} from the uniform buy-at-bulk problem to the general robust network design ($\optspr$).
Andrews~\cite{Andrews,AndrewsPC} showed that under a plausible complexity theoretic assumption ($NP \not\subseteq ZPTIME(n^{\mathrm{polylog}~n})$), there is
no polytime algorithm for uniform buy-at-bulk with approximation guarantee within $O(\log^{1/4-\epsilon}(n))$, for any $\epsilon >0$. These two results
imply that the gap between $\optmpr,\optspr$ must be similarly large, since $\optmpr$ is polytime computable, and could otherwise be used
to approximate (the decision form of) uniform buy-at-bulk.
We will demonstrate all the other gaps in this paper; most of the work is on the most interesting case, between $\optdyn$ and $\optmpr$.




\paragraph{Discussion.} In the robustness paradigm, the question of how large these gaps
can be is asked for specific classes of demand polyhedra. A class
that has received much attention consists of the so-called ``hose
models'' which come in symmetric and asymmetric flavours. In the
symmetric hose model, each terminal $v$ has an associated
\emph{marginal} $b_v$, which represents an upper bound on the
\emph{total} amount of traffic that can terminate at $v$. The demand
polytope consists of all symmetric demands which do not violate
these ``hose'' constraints; i.e. $\sum_j D_{ij} \leq b_i$ for each
terminal $i$. The asymmetric hose problem is similar, but the
terminals are divided into sources and sinks; all demand is between
source and sink nodes, and again, total demand to or from a terminal
cannot exceed its marginal. Classes such as the hose model arise
naturally in switch design, but they were also motivated by
applications in data networks \cite{fingerhut97,guptavpn01}; one of
these is referred to as the {\em virtual private network} (VPN)
problem.

It is implicit in Fingerhut et al.~\cite{fingerhut97} and explicit in Gupta et
al.~\cite{guptavpn01} that in the {\em symmetric} hose model,
$\optmpr \leq \optspr \leq 2\optdyn$. However, the gap instance
between $\optmpr$ and $\optdyn$ that we demonstrate in this paper is
in fact an instance of the asymmetric hose problem, and hence there
is a logarithmic gap for this latter model.\footnote{This rectifies
an earlier assertion (cf. Theorem~4.6 in \cite{chekurisurvey07}).}
We describe a class of graphs $G$, cost function $c$, and a demand
polytope $\Po$, such that $\optdyn(G,c,\Po)=O(n)$ but
$\optspr(G,c,\Po)=\BigOmega(n\log n)$ and
$\optmpr(G,c,\Po)=\BigOmega(n\log n)$. The polytope $\Po$ has the
property that all demands share a common ``sink'' node.


It turns out that the problem of designing an \SPR\ routing template
for our gap instance corresponds to the well-studied rent-or-buy
network flow problem, which is a
generalization of the Steiner tree problem. In this problem there is
only one demand matrix instead of a polytope of demands, but the
cost function is concave. We sketch the lower bound argument for
$\optspr$ separately in Section~\ref{sec:spr} since it is much
simpler; it proceeds by showing that the optimal \SPR\ templates may
be assumed to be tree templates for our gap instance.

The lower bound for $\optmpr$ is more involved.  We show that the cost of an \MPR\ template for our
gap instance can be characterized by a network design problem that we call \emph{buy-and-rent}.  Again
there is only one demand to be satisfied, but the cost function is more complex.
The buy-and-rent cost function seems to be new and natural: briefly, instead of asking that each edge be
either rented or bought, it allows that capacity may be partially bought and the rest rented.
This new cost function is more amenable to analysis, and leads to our lower bound for $\optmpr$.

\paragraph{Relation to congestion lower bounds.} We remark that our lower bounds for
the total cost model also imply lower bounds for minimizing the
maximum congestion, essentially because if every edge had congestion
at most $\alpha$, the total cost would also be bounded by a factor
$\alpha$. Since the polytope $\Po$ we use is a subset of the
single-sink demands routable in $G$, this also implies a result in
\cite{haji} which gives an $\BigOmega(\log n)$ bound for congestion
via oblivious routing of single sink demands (although their
analysis also extends to the case of lower bounding performance of a
general online algorithm). Congestion minimization problems can be
seen as equivalent to a robust optimization where one uses maximum
edge congestion as a cost function;    simply take the polytope
consisting of {\em all} single-sink demands which are routable  in
$G$ (this is a superset of our choice $\Po$). The construction in
\cite{haji} uses meshes (grids), building on work of
\cite{BartalL99,Maggs97}.
This construction does not seem to extend to the total cost model however, and
we use instead a construction based on expanders, extending and
simplifying a connection shown in earlier work
 \cite{ChekuriHardness07}.


\paragraph{Gaps for tree routing.} In Section~\ref{sec:trees} we give a family of instances (using a different demand polyhedron) showing that
the gap between $\optspr$ and $\optT$ can be $\BigOmega(\log n)$.
This immediately implies that the gaps between $\optdyn$ and $\optT$ and between $\optmpr$ and $\optT$ is $\BigOmega(\log{n})$ for this family of instances.

\section{A gap example}

\subsection{A robust network design instance}\label{sec:instance}

Let $G=(V,E)$ be a graph on $n$ nodes with constant degree $d \geq 3$ and edge-expansion at least $1$;
i.e. we have that $|\delta_{G}(S)| \geq |S|$ for all $S \subseteq V$
with $|S| \leq n/2$.  Here $\delta_{G}(S)$ denotes the set of edges in $E$ with one end-point in $S$
and the other outside $S$.  It is well-known that such edge-expanders with the above parameters exist.
Now add a special sink node $r$ to $V$ to obtain our instance $\bar{G}=(\bar{V}, \bar{E}) = (V \cup \{r\}, E \cup \{vr: v \in V\})$; see Figure~\ref{fig:instance}.

We look at a single-sink hose model (cf. \cite{ChekuriHardness07}),
where our demands come from a polytope $\mathcal{P}$ defined as
follows. We have a specified marginal capacity $b_v$ at each node:
$b_r  ~= \beta n$ (where $0 < \beta < 1$), and $b_v=1$ for all $v
\in V$. Each demand matrix $D_{ij} \in \mathcal{P}$ has the property
that $\sum_{j} D_{ij} \leq b_i$ for each node $i \in \bar{V}$, and
$D_{ij} > 0$ only if $r \in \{i,j\}$.  Although we  often think of
nodes routing flow towards the sink,  the demands and flows are undirected in
this paper.

\begin{figure}
    \centering
    \includegraphics{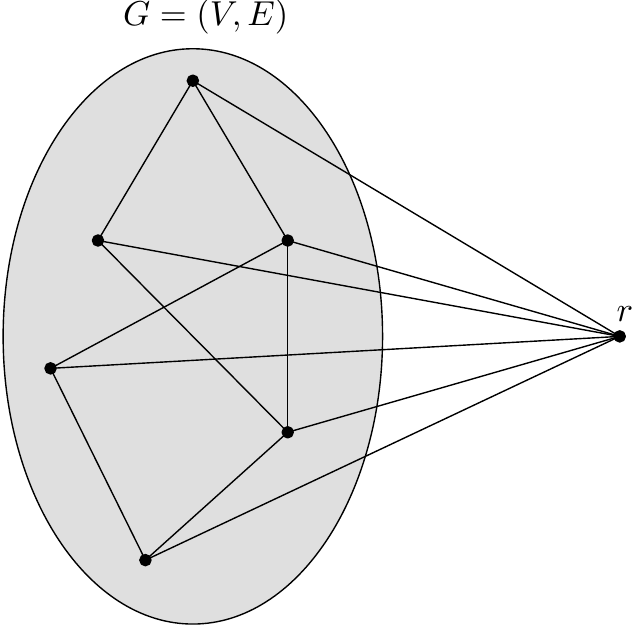}
    \caption{The gap instance. $G$ is a $d$-regular expander}\label{fig:instance}
\end{figure}

Thus each demand matrix we must
support, identifies a single-sink network flow problem. It is a
simple exercise to see that:

\begin{lemma} \label{lemma_br}
If $b_r$ is an integer, then our network
is robust for $\mathcal{P}$ and a given routing model if and only if for each subset $X$ of $b_r$ nodes in $G$,
there is enough capacity to route one unit from each node in $X$
to $r$, using the prescribed routing model.
\end{lemma}
We  use this fact below. Finally, we also assign costs to the edges:
each edge of $G$ has cost $1$, and each edge in
$\delta_{\bar{G}}(r)$ has cost $1/\beta$.

Our main result is the following theorem:
\begin{theorem}
\label{thm:gap}
For $\beta = 1/\log n$, there is a dynamic routing for the single-sink hose model instance (defined
above) of cost $O(n)$, but every
\MPR\ solution (and hence every \SPR\ solution) has cost $\BigOmega(n \log n)$.
\end{theorem}


The first assertion is proved in the next section. In Section~\ref{sec:spr},
we see that determining $\optspr$ for single-sink hose models is equivalent
to the well-studied \emph{single-sink rent-or-buy} problem, 
and the rent-or-buy  problem always has a tree solution.
This can be used to show that $\optspr = \BigOmega(n \log n)$ for our instance with $\beta=1/\log{n}$.
We give a sketch of a proof of this since it is considerably simpler than (but implied by) the proof of the corresponding bound for \MPR.
This \MPR\  lower bound is demonstrated in Section~\ref{sec:mpr}.

We assume throughout the paper that $b_r = \beta n$ is an integer.

\subsection{A solution for the dynamic routing model}
Put capacity $\beta$ on each edge of $\delta_{\bar{G}}(r)$, and
capacity $1$ on each edge of $G$. Clearly, the cost of this
reservation is $O(n)$ independent of $\beta$. We  show that this is
a valid \FR\ capacity reservation. Using Lemma~\ref{lemma_br} it
suffices to show that for any subset of $\beta n$ nodes $X$ in $G$,
all nodes in $X$ can simultaneously route a unit flow to $r$.  To
this end, we add a new node $t$ to $\bar{G}$ and edges $vt$ for
$v \in X$ with unit capacity to form graph $G'$. We  show that $G'$
supports a $t$-$r$ flow of size $|X|=\beta n$.  By the max-flow
min-cut theorem it suffices to show that all $r$-$t$ cuts in $G'$
have size at least $\beta n$, i.e. that for each $S \subseteq V$ we
have $|\delta_{G'}(S \cup \{t\})| \geq \beta n$.

We have
\[ |\delta_{G'}(S \cup \{t\})| = \beta |S| + |X \setminus S| + |\delta_G(S)|. \]
Now, if $|S| \leq n/2$ then using the fact that for $G$ we have $|\delta_G(S)| \geq |S|$ we get
\begin{align*}
|\delta_{G'}(S \cup \{t\})|  &\geq \beta |S| + |X \setminus S| + |S| \\
 &\geq \beta |S| + |X|\\ &\geq  |X|.
\end{align*}

And if $|S| > n/2$ then using the fact that for $G$ we have $|\delta_G(S)| \geq n-|S|$ we get
\begin{align*}
|\delta_{G'}(S \cup \{t\})| & \geq \beta |S| + |X \setminus S| + n-|S| \\
& \geq  \beta |S| + |X \setminus S| + \beta(n-|S|)  \\
& =  \beta n + |X \setminus S| \\
& \geq  \beta n.
\end{align*}

%
%

Hence the above capacity reservation can support the \FR\ routing model and costs $O(n)$.

\subsection{Rent-or-buy: an $\BigOmega(\log n)$ gap between \FR\ and \SPR}
\label{sec:spr}

Note that the optimal cost oblivious \SPR\ network can
be cast as a minimum cost (unsplittable) flow problem as follows.
Each node $v \in V$ must route one unit of flow on a path $P_v$ to $r$ and the overall
(truncated) cost of path choices is: $\sum_{e} c(e)
\min\{N(e),b_r\}$, where $N(e)$ is the number of nodes $v$ whose
path to $r$ used the edge $e$. Clearly, if the capacity of each edge is
$\min\{N(e),b_r\}$, then we have sufficient capacity to route any
demand matrix in $\mathcal{P}$ using as a template the paths $P_v$.
The converse is in fact also true and easy; any template gives rise to a
corresponding integer flow whose truncated cost is the same.

This truncated routing cost problem is simply a so-called single-sink  {\em
rent-or-buy} (\ssrob) problem (see e.g. \cite{Grandoni08Core,GuptaKPR07}).
We are given a network $G$ with edge costs $c(e)$, and a special sink node $t$.
A parameter $B \geq 1$ is also given (this will equal $b_r$ in the instance corresponding to \SPR).
We also have a list of sources $s_i$ for $i=1,2 \ldots, p$; each
source needs to route to the sink $t$. For each edge in the network,
we may either purchase it at a cost of $Bc(e)$, in which case it is
deemed to have infinite capacity, or otherwise we may rent it. In that case,
we must pay $c(e)$ per unit of capacity that we use on the edge. The
goal is to find which edges to buy and which to rent in order to
support a flow from each node to $t$, at the smallest possible cost.
In other words, we seek a fractional flow $\vec{f}$ of the demands
that minimizes $\sum_{e\in E} c(e)\min\{f(e), B\}$ (we will see next that the optimal solution will always be integral, ensuring that we do in fact have a correspondence with \SPR). In general, we
may also consider such \emph{single-sink flow problems with concave
costs} $\sum_e  g_e(f(e))$ where each $g_e$ is a concave function.

The following result is immediate from the concavity of the cost function:
\begin{proposition}
\label{prop:tree}
Any single-sink flow problem with nondecreasing concave costs has an optimal solution whose support is a tree.
In particular, such an optimal solution always exists for the \ssrob\ problem.
\end{proposition}
\iftrue
\begin{proof} 
    Let $\vec{f}$ be a flow giving an optimal solution to the flow problem, chosen so that $\supp(\vec{f})$ is setwise minimal.
We show that then $\supp(\vec{f})$ must form a tree.


Let us consider $\vec{f}$ as a directed flow, where each terminal sends flow
to the sink. If there is any directed cycle in the support of $\vec{f}$, then we
may simply reduce flow on this cycle until some arc becomes zero;
this does not increase the cost since our cost function is
nondecreasing.
So we may assume our support is acyclic in
the directed sense.
Suppose now that there is some undirected
cycle $C$ in the support which by assumption corresponds to some  forward (traversing $C$ in order) arcs $F$
and some reverse arcs $R$. Let $\epsilon = \min\{f(a): a \in R \cup F \}$. 
Define two solutions $\vec{f}^+,\vec{f}^-$ by $f^{\pm}(a) = f(a) \pm \epsilon$ for $a \in F$, and $f^{\pm}(a)=f(a) \mp \epsilon$ for $a \in R$. 
By concavity, $C(\vec{f}) \geq (1/2) [C(\vec{f}^+)+C(\vec{f}^-)]$.
Then since $\vec{f}$ was an optimal solution, $C(\vec{f}^+) = C(\vec{f}^-) = C(\vec{f})$. Hence both $\vec{f}^+$ and $\vec{f}^-$ are optimal, and one of them must have smaller support than $\vec{f}$, a contradiction.
\qed
\end{proof}
\fi

Note that the preceding result shows that in the case of
single-sink hose models, $\optspr=\optT$.
It is not the case that $\optmpr=\optT$ in this setting however.  If that were the case, \ssrob\ would be polynomially solvable,
but the case where $b_r=1$ already captures the Steiner Tree problem.
Because of this tree structure, arguing why the gap holds in the case of \SPR\ is considerably simpler.
The argument contains some intuition as to why the gap also holds for \MPR, so we describe the approach now.

By Proposition~\ref{prop:tree}, we may represent the optimal \SPR\ solution with a tree $T$, which we think of as being rooted at $r$.
First let us suppose that the solution uses only one edge $rv$ from $\delta(r)$, so that all terminals must route via $v$.
Since $G$ was bounded degree this means that many nodes (a constant
fraction) must use long paths, of length $\log_d (n)$.
If these all had to pay one unit along their whole path, then this
already costs $\BigOmega ( n \log n)$.
But it is not as easy as that; if we have a subtree $T_w$ rooted at node $w$ that contains at least $b_r = \beta n$ nodes, then in fact we only need to pay for $b_r$ units on the edge out of $w$.

Imagine removing the edges of $T$ which are used by more than $\beta
n$ terminals, leaving a number of subtrees, each containing at most
$\beta n$ terminals. If $T$ is fairly balanced, there are around
$\BigTheta(n/(\beta n)) = \BigTheta(1/\beta)$ such subtrees. (If $T$
is very unbalanced, there could be many more---consider a
caterpillar. For the full proof, one must use the larger distances
of leaves to the root to get the required bound.)
%
Roughly speaking, in each such subtree, a good fraction of the leaves are a distance roughly $\log \beta n$ from the root of this subtree.
Since there is no cost sharing within this subtree, these nodes really do pay $\beta n \log (\beta n)$.
Thus the subtrees combined pay \[ \BigOmega\left(1/\beta \cdot \beta n \log (\beta n) \right) = \BigOmega\left(n\log (\beta n)\right). \]
If we set $\beta = \frac{1}{\log n}$, this yields a cost of $\BigOmega(n \log n)$.

To make the above argument precies, we must balance the use of multiple edges into $r$.
Label the children of $r$ in $T$, $1$ through $m$, and let $k_i$ be the size of subtree $i$. Let $L$ be the set of \emph{heavy} children of $r$ in $T$, i.e., $k_i > \beta n$; let $R$ be the set of \emph{light} children of $r$.

Suppose $i$ is a heavy child. The subtree $T_i$ routed at $i$ has some set of heavy edges $H$, i.e., edges with flow more than
$b_r = \beta n$; let $p = |H|$. Now consider the tree $T_i'$ obtained from $T_i$ by contracting the edges in $H$. The root of $T_i'$ has degree at most $pd$; all other nodes have maximum degree $d$. The maximum number of nodes that are a distance less than $j$ from the root is
\[ \sum_{i=1}^j pd^i = pd(d^j-1)/(d-1). \]
Taking $j=\log_d (k_i/10p)$, the above is a constant fraction of the nodes, and the rest must be further away. So a constant fraction of the nodes are a distance $\BigOmega(\log(k_i/p))$ away from the root of $T_i'$. Since the edges in $T_i'$ are not heavy, these nodes contribute
\[\BigOmega(k_i\log (k_i/p)) = \BigOmega(k_i\log(\beta n/p))
\]
 to the cost of tree $T_i$. Adding the cost of the heavy edges, we get a total cost of
\[ \BigOmega(p\cdot \beta n+ k_i\log(k_i /p)). \]

We verify that this is at least $\BigOmega(k_i\ln (\beta n))$. It suffices to show that
$p\cdot \beta n+ k_i\ln(k_i /p) \geq k_i\ln (\beta n)$.
But this is equivalent to
\[ \frac{p \beta n}{k_i} \geq \ln \left(\frac{p \beta n}{k_i}\right), \]
which is clearly true since $x \geq \ln x$ for all $x > 0$.

For $i \in R$, a lower bound can be obtained by considering only the edge $ir$, which contributes $k_i / \beta$ (remembering that edges adjacent to $r$ have cost $1/\beta$). So we have the following lower bound on the cost of the tree solution:
\[ \sum_{i \in R} k_i/\beta + \sum_{i \in L} k_i\log(\beta n) / C, \]
where $C$ is some constant. This is at least
\[ n\cdot \min (1/\beta, \log(\beta n)/C). \]
Setting $\beta = 1/\log n$ gives the result.

\subsection{Buy-and-rent: an $\BigOmega(\log n)$ gap between \FR\ and \MPR}
\label{sec:mpr}

The main difficulty in analyzing the \MPR\ model is 
that we can no longer restrict to tree like routings as we could for \SPR; there is no equivalent of Proposition~\ref{prop:tree} for \MPR.
In particular, the \MPR\ problem for our instance is not captured by a
\ssrob\ problem. Instead, we get a new kind of routing cost model as explained
below.

Let us first examine more closely the cost on edges induced by an \MPR\ routing
template for a single-sink hose design problem.  As in Lemma~\ref{lemma_br}, it is sufficient to consider the cases
where we wish the network to support the routing of any $\beta n$ of the nodes in $V$ to the sink $r$ simultaneously.
Suppose that $f_i(e)$ represents the flow that node $i$ sends on edge $e$ in a template, then for the single
sink hose design problem, the formula for the capacity needed by $e$ is:
\begin{align} \label{eqn:mprcost}
\max_{D \in \mathcal{P}} \sum_{i \in V} D_{ir} f_i(e)  = \max_{W \subseteq V: |W|=\beta n} \sum_{i \in W} f_i(e),
\end{align}
where $\mathcal{P}$ is the set of single-sink hose matrices.
In other words, the capacity needed on edge $e$ is just the sum of the $\beta n$ largest values of $f_i(e)$.

We introduce a new routing cost model which we call (single-sink) {\em
buy-and-rent} ({\BaR}). This  exactly models the \MPR\ cost model
defined above, but is more manageable in terms of analysis. In the
buy-and-rent problem, there are costs on the edges, and unit demands
from some subset $W$ of nodes called {\em terminals}. Each terminal
wishes to (fractionally) route one unit  of demand to the sink $r$.
Apart from the costs $c(e)$ on the edges, we also have a parameter
$k$.
The difference from rent-or-buy is that we may now purchase some capacity
amount $\gamma(e) \in [0,1]$ (in rent-or-buy we would buy an infinite capacity
link) and the interpretation is that every terminal is allowed to use up to
$\gamma(e)$ units of capacity on the edge. If it chooses to route any more on
that edge, then it must pay for the additional rental cost.
The cost of purchasing capacity on an edge $e$ is $k \gamma(e) c(e)$.

Buy-and-rent can be considered as an LP relaxation of (single-sink) rent-or-buy; this formulation is in fact very similar to the LP relaxation used by Swamy and Kumar~\cite{SwamyK04} to give constant factor approximation algorithms for connected facility location and single-sink rent-or-buy.
Their formulation is stronger however (in that the optimum for their LP lies between the \BaR\ and \SPR\ optima), and so does not exactly model the \MPR\ problem.
In particular, in buy-and-rent, solutions may conceivably use flow paths that alternate several times between rented capacity and purchased capacity.
In contrast, a solution to the LP of Swamy and Kumar~\cite{SwamyK04} always has a connected ``core'' of purchased edges containing the sink node and terminals use rented capacity to route to that core.

\begin{proposition}
\label{prop:equiv}
The buy-and-rent problem with parameter $k = \beta n$, and the single-sink hose design problem in the
\MPR\ routing model have the same optimal solution.
\end{proposition}
\begin{proof}
Suppose that $(\vec{f}_i)$ is an \MPR\ routing template for the robust
hose design problem. Consider the {\BaR} solution for parameter $k = \beta n$
obtained as follows. For each edge $e$, order the terminals so that
$f_{\pi(1)}(e) \geq f_{\pi(2)}(e) \geq \ldots f_{\pi(n)}(e)$. Then
we purchase $\gamma(e) = f_{\pi(k)}(e)$ units of capacity on edge
$e$, and we use the same routing $\vec{f}_i$ as the \MPR\ solution. This
guarantees that for any edge, none of the terminals $\pi(j)$ with $j
> k$,  pays to route on edge $e$ since we purchased enough
capacity for them to travel for free. For each terminal $\pi(j)$
with $j \leq k$, it must pay the rental cost to route
$f_{\pi(j)}(e)-f_{\pi(k)}(e) \geq 0$. This costs $c(e)$ times the
amount $\sum_{j \leq k} (f_{\pi(j)}(e)-f_{\pi(k)}(e)) = \sum_{j \leq
k} f_{\pi(j)}(e) - k f_{\pi(k)}(e)$. Since the purchased capacity
cost $ k f_{\pi(k)}(e) c(e)$, the total buy-and-rent cost is $c(e)
\sum_{j \leq k} f_{\pi(j)}(e)$ which is the cost of edge $e$ in the
\MPR\ template using \eqref{eqn:mprcost}.

Conversely, suppose that we have a minimum cost solution for {\BaR}
and consider the robust design cost for using the same routing as a template.
Without loss of generality $\gamma(e) = f_{\pi(k)}(e)$ since if
$\gamma(e)$ was larger than this, then by reducing the capacity
bought by sufficiently small $\epsilon > 0$, the rental costs are
unaffected for terminals $\pi(j)$ for $j \geq k$. And for terminals
$\pi(j)$ with $j < k$, their rental cost increases by at most
$\epsilon c(e)$. Hence the total rental cost increases by $(k-1)
\epsilon c(e)$, and the total cost of bought capacity reduces by $k
\epsilon c(e)$, thus decreasing the overall cost.

Similarly, if $\gamma(e) < f_{\pi(k)}(e)$, then increasing the
bought capacity $\gamma$ by some small $\epsilon > 0$, has cost of
$k \epsilon c(e)$.  But the reduction in rental costs is at least
the reduction in rental cost of the first $k$ terminals which is $k
\epsilon c(e)$, and thus the overall cost does not increase as a
result of increasing $\gamma$. Hence the cost of edge $e$ is just
the purchase cost $c(e) \cdot  k f_{\pi(k)}(e)$ plus the rental cost
$c(e) \sum_{j \leq k} (f_{\pi(j)}(e)-f_{\pi(k)}(e))$ and this is
identical to the robust design cost when using the same template.
\qed
\end{proof}

We again take $\beta = 1/\log n$ (so $k = n / \log n$).
We now prove that any solution to the \BaR\ problem on our expander
instance is expensive; this together with the preceding proposition
implies our main result, Theorem~\ref{thm:gap}.
\begin{theorem} Any solution to the \BaR\ problem
with $\beta = 1/\log{n}$ on the expander instance (defined in Sec.~\ref{sec:instance}) has cost $\BigOmega(n \log n)$.
\end{theorem}
\newcommand{\gammaP}{\gamma(\delta(r))}
\newcommand{\gammaE}{\gamma(E)}
\newcommand{\gammaBE}{\gamma^E}
\newcommand{\gammaBP}{\gamma^P}
\newcommand{\Gclosure}{\bar{G}}
\newcommand{\flowr}{\mu^r}
\newcommand{\flowb}{\mu^b}
\newcommand{\flowt}{\mu^t}
\newcommand{\Crent}{C(\text{rent})}
\newcommand{\Rad}{R}
\begin{proof}
Consider an arbitrary \BaR\ solution, determined by bought capacity $\gamma_e$ on each edge, and a flow template
$(\vec{f}_i: \mbox{ for each terminal $i$})$.

For a set $A$ of edges, let $\gamma(A) := \sum_{e \in A} \gamma_e$. Thus
$\gammaP := \sum_{v \in V} \gamma_{vr}$ is the total bought capacity on the \emph{port edges} (these are the edges
connecting $r$ to the nodes in $V$), and $\gammaE := \sum_{e \in E} \gamma_e$ is the capacity bought in the expander.
The cost of buying capacity in the expander is then $k \cdot \gammaE$, so we may assume that $\gammaE < \log^2n$,
or else the solution already costs $\BigOmega(n \log n)$.
A similar argument for port edges (but recalling that these edges cost $\log n$) allows us to assume that $\gammaP < \log n$.

    For a terminal $v$, let $B_i(v)$ be the set of nodes (or sometimes, their induced graph)
in the expander that are a distance at most $i$ from $v$.
    We are particularly interested in balls of radius $\Rad := \lfloor \log_d \sqrt{n}  \rfloor  - 1 =
\lfloor \log n / (2 \log d) \rfloor  - 1$;
    we  use $B(v)$ as shorthand for $B_\Rad(v)$.
    Note that since $G$ is $d$-regular,
    \[ |B(v)| \leq \sum_{i=0}^{\Rad} d^i \leq d^{\Rad +1} \leq n^{1/2}. \]

    Let $\gammaBE(v) := \sum_{e \in E : e \subset B(v)} \gamma(e)$ and $\gammaBP(v) := \sum_{w \in B(v)} \gamma(wr)$.
    A single $\gamma(e)$ for an edge $e=u_1u_2$ contributes to many $\gammaBE(v)$'s, but not too many:
    \[ |\{v : e \subset B(v)\}| \leq |\{v : u_1 \in B(v) \}| = |B(u_1)| \leq n^{1/2}.  \]
    So we must have that
    \begin{equation}\label{eq:gammaBE}
        \sum_{v \in V} \gammaBE(v) \leq n^{1/2}\gammaE \leq n^{1/2}\log^2 n.
    \end{equation}
    Similarly,
    \begin{equation}\label{eq:gammaBP}
        \sum_{v \in V} \gammaBP(v) \leq n^{1/2} \log n.
    \end{equation}

Consider an arbitrary terminal $v$. The unit of flow from $v$ can be divided up into three types depending on how
the flow enters $r$:
\begin{itemize}
    \item A fraction $\flowr_v$ of flow that rents on the port edge it uses.
    \item A fraction $\flowb_v$ of flow that uses bought port capacity, on a port within a distance $\Rad$ from $v$.
    \item A fraction $\flowt_v$ representing all remaining flow; this flow must ``travel'' and use port edges that are further than $\Rad$ from $v$.
\end{itemize}
Clearly $\flowr_v + \flowb_v + \flowt_v = 1$.

We now aim to find a lower bound on the total rental cost paid by the terminals.
Flow that rents the port edge must pay $\log n$ just for this edge, giving a cost of $\flowr_v \log n$.
Now consider the $\flowt_v$ fraction of flow that travels outside the ball $B(v)$ in the expander before using a port edge.
This flow must cross each of the cuts $C_i := \delta(B_i(v))$, for $0 \leq i \leq \Rad$.

The maximum amount of flow that can travel across cut $C_i$ for free (using the bought capacity) is $\gamma(C_i)$, and so there is a rental cost of at least $\flowt_v - \gamma(C_i)$ in crossing cut $C_i$.
Summing over all the cuts, we find that the rental cost associated with this travelling flow is at least
\[ \sum_{i =0}^{\Rad - 1} (\flowt_v - \gamma(C_i)) \geq  \Rad \flowt_v - \gammaBE(v). \]
Thus the rental cost associated with terminal $v$ is at least
\[ \log n \cdot \flowr_v + \Rad \flowt_v - \gammaBE(v). \]
Summing this over all terminals $v$, we obtain a total rental cost of at least
\begin{align*}
    \Crent &\geq \sum_{v \in V} (\log n \cdot \flowr_v + \Rad \cdot \flowt_v) - \sum_{v \in V} \gammaBE(v)\\
    &\geq R \sum_{v \in V} (\flowr_v + \flowt_v) - \sum_{v \in V} \gammaBE(v) & & \text{since $R \leq \log n$}\\
    &\geq R \sum_{v \in V} (\flowr_v + \flowt_v) - n^{1/2} \log^2 n & &  \text{by \eqref{eq:gammaBE}}.
\intertext{Finally, note that}
    \sum_{v \in V} (\flowr_v + \flowt_v) &= \sum_{v \in V} (1 - \flowb_v) \geq \sum_{v \in V} (1 - \gammaBP(v))\\
    &\geq n - n^{1/2}\log n & &  \text{by \eqref{eq:gammaBP}}.
\intertext{Thus}
\Crent &\geq R \cdot (n - n^{1/2} \log n ) - n^{1/2} \log^2 n \\
&= \BigOmega(n \log n), 
\end{align*}
since $R = \BigTheta(\log n)$. \qed
\end{proof}

\section{Single path routing vs. tree routing}
\label{sec:trees}

As discussed in the introduction, 
for any robust network design problem we have $\optT = \tilde{O}(\log{n}) \optdyn$.
We now show that this is essentially best possible by exhibiting a problem instance such that
$\optT = \BigOmega(\log{n}) \optspr$, and so also $\optT = \BigOmega(\log{n}) \optdyn$.

Consider a graph on $n$ vertices with girth (length of the shortest
cycle in the graph) $\BigOmega(\log{n})$, and with $cn$ edges, where
$c$ is a constant strictly larger than $1$. The requirement $c>1$ makes the existence of such graphs nontrivial, but they do exist. For example, Lemma~15.3.2 in \cite{Matousek} states that there exist graphs with girth $\ell$ and $\frac{1}{9} n^{1+1/(\ell-1)}$ edges.
Taking $\ell := (\log{n})/100$ gives a graph satisfying our requirements.

This graph defines the network topology for our problem instance: all the nodes are
terminals, and all edges have unit cost.  The demand polytope is
given by a single demand: there is a unit demand between terminals
connected by an edge.

Clearly, a good \SPR\ template is the network itself, and its cost is $cn$,
the number of edges.  Now, if we take any tree
template, then edges of the network that are not included in the
tree have to be routed on a path of length $\BigOmega(\log n)$ because of the girth property. There are at least $cn-(n-1) = (c-1)n+1$ such edges, and so the total cost
of the tree template is $\BigOmega(n \log n)$.

\section{Conclusions}


We have shown that oblivious routing (even splittable) can perform quite poorly compared to dynamic routing in some situations.
However, fully dynamic routing is problematic to implement.
Is it possible that some tradeoff between the two extremes of dynamic and oblivious routing could produce significantly better results while remaining practical?

Another very natural question concerns the gap between \MPR\ and \SPR\ for the single-sink robust network design problem with arbitrary demand polytopes.
We are not aware of any single-sink instances for which the gap is superconstant.
The single-sink robust design problem (computing \optspr) could still conceivably have a constant factor approximation algorithm for well-described polytopes.
This would be of interest since it generalizes a host of well-known problems such as Steiner tree, single-sink rent-or-buy, and single-sink buy-at-bulk (the last
follows from a transformation given in~\cite{OlverS10}).

\vspace{1ex}
\noindent \textbf{Acknowledgements.}
We would like to thank Gianpaolo Oriolo for some very helpful
discussions. We also thank an anonymous reviewer for ESA, for some
detailed and useful input. A final version of this paper was published in Algorithmica.

The second author was supported by a MELS Quebec Merit Scholarship for Foreign Students, and a Schulich Fellowship.
The third author is supported by a NSERC Discovery Grant.

\bibliographystyle{abbrv}    
\bibliography{vpn}

\iflong
\appendix
\section{Upper bounds on the gaps}
For the sake of completeness, in this appendix we give a proof of Gupta's observation that the gap between
$\optdyn$ and $\optspr$ is $O(\log{n})$.  We also show that the gap between $\optdyn$ and $\optT$ is
$\tilde{O}(\log{n})$, via a similar proof.  A sketch proof of Gupta's observation appears in
Chekuri~\cite{chekurisurvey07}.

We use basic notions about finite metric spaces; an excellent
exposition of this topic is found in Matousek's book~\cite{Matousek}.  We begin
with some notation and state a theorem that we need. We are given an
instance $(G, c, \mathcal{P})$ of the robust network design problem
on $n$ nodes.
The cost function $c$ induces a metric $d_G(\cdot, \cdot)$ on nodes of
$G$ in the usual way: the distance $d_G(x,y)$ between nodes $x$ and $y$ is given by the length of the shortest $x$-$y$
path in $G$, with edge $e$ having length $c(e)$.
We also define the complete graph $C_G$ on $V$ where edge $\{x,y\}$ has length $d_G(x,y)$.

Now using the result and notation of \cite{FRT04}, a metric $d$ can
be approximated by distribution over dominating tree metrics in the
following sense. A metric $(V, d')$ is said to \emph{dominate} a
metric $(V,d)$ if for all $x, y \in V$ we have $d'(x,y) \geq
d(x,y)$.  Given a probability distribution $\mathcal{D}$ over a
family of  tree metrics $\mathcal{S}$ on $V$, we say that
$(\mathcal{S}, \mathcal{D})$ $\alpha$-probabilistically approximates
a metric $(V, d)$ if every metric in $\mathcal{S}$ dominates $d$,
and for all $x,y \in V$ we have $\expect_{d' \in (\mathcal{S},
\mathcal{D})}[d'(x,y)] \leq \alpha \cdot d(x,y)$.

Building upon previous work, Fakcharoenphol et al.~\cite{FRT04}
proved that every finite metric on $n$ nodes can be
$O(\log{n})$-probabilistically approximated by a distribution over
tree metrics.

Now getting back to our robust network design instance, we find a
distribution $(\mathcal{S}, \mathcal{D})$ over tree metrics which
$O(\log{n})$-probabilistically approximates $d_G$. Trees in
$\mathcal{S}$ can be taken to be spanning trees of $C_G$ (they need
not be subtrees of $G$ though). For a capacity reservation $u$ on
the edges of $G$, its cost is $\cost_G(u) := \sum_{e \in E}c(e)
u(e)$. We can also define the cost of this reservation on a tree
metric $T$ by $\cost_T(u):= \sum_{e \in E} d_T(e) u(e)$. Let $u^*$ be
the optimum capacity reservation for the \FR\ routing model, so
$\cost_G(u^*)=\optdyn$. By the theorem of \cite{FRT04} we have by
linearity of expectation applied to $\cost_T(u^*)$:
\begin{align*}
\cost_G(u^*) \leq \expect_{d' \in (\mathcal{S}, \mathcal{D})} \cost_T(u^*) \leq O(\log{n}) \cost_G(u^*).
\end{align*}

So there exists a tree $T \in \mathcal{S}$ such that
\begin{align*}
\cost_G(u^*) \leq  \cost_T(u^*) \leq O(\log{n}) \cost_G(u^*).
\end{align*}

Let $u_T$ denote the optimal capacity vector for the robust network
design problem on the graph $T$\footnote{Computing the capacity of
an edge $e \in T$ amounts to solving a linear program over
$\mathcal{P}$ with objective $\sum_{i \in A,j \in B} D_{ij}$ where
$A,B$ are the two components of $T-e$.}; N.B. all routing models are
equivalent on a tree $T$. We also have $\cost_T(u_T) \leq
\cost_T(u^*)$. This is because the dynamic solution gives an
oblivious solution for $T$ with cost at most $\cost_T(u^*)$ as
follows. For any edge $e$ in $G$, add $u^*(e)$ units of capacity on
the path in $T$ between the endpoints of $e$. The overall capacity
$u'$ installed then costs $\cost_T(u^*)$. For any valid demand, the
dynamic solution satisfied the demand by assigning flows $f(P)$ to
paths in $G$. Moreover, $\sum_{P: e \in P}  f(P) \leq u^*(e)$ and
thus by definition, $u'$ has enough capacity to support routing all
such flow paths  $P$, between $u,v$ say, on the unique $u-v$ path in
$T$. Since $\cost_T(u_T) \leq \cost(u')$ we are done. Hence
\begin{align*}
\cost_G(u^*) \leq  \cost_T(u_T) \leq O(\log{n}) \cost_G(u^*).
\end{align*}

Let $f_T$ be the routing template  on $T$ that determines $u_T$.
This can be transferred to an \SPR\ routing template in $G$ with a
capacity reservation $u_{G(T)}$ on $G$ with the same cost as
follows: Each edge in $T$ corresponds to a path in $G$.  For edge
$xy$ in $T$ we reserve $u_T(xy)$ capacity on each edge on the path
in $G$ corresponding to edge $xy$.  If an edge in $G$ lies on
several such paths then the capacity reserved on it is the sum of
the $u_T$-values for all of these paths. Clearly the cost of the
resulting capacity reservation $u_{G(T)}$ on $G$ is the same as the
cost of $u_T$.  Also, the routing template $f_T$ for $u_T$ can be
simulated on $G$ in the natural fashion in the \SPR\ routing model.
Thus, we have a reservation $u_{G(T)}$ supporting a \SPR\ routing on $G$
and with $\cost_G(u_{G(T)}) = \cost_T(u_T) \leq
 O(\log{n}) \cost_G(u^*)$.
Noting that $\optspr \leq \cost_G(u_{G(T)})$ and $\optdyn =
\cost_G(u^*)$ completes the proof that $\optspr = O(\log{n}) \optdyn$.

Note that the above proof does not give us that $\optT = O(\log{n})
\optdyn$ because the support of $u_{G(T)}$ need not be a tree. We
can prove a slightly weaker result, namely $\optT =
\tilde{O}(\log{n})\optdyn$ by invoking a theorem of
\cite{AbrahamBN08}: For any metric $d_G$ induced by a graph $G$ on
$n$ nodes there is a distribution on the \emph{spanning trees} of
$G$ which $\tilde{O}(\log{n})$-probabilistically approximates $d_G$.
The remaining details are essentially the same as in the above,
except that this time the support of $u_{G(T)}$ is indeed a tree.
\end{document}

\fi

\end{document}